\documentclass[conference]{IEEEtran}
\usepackage{graphicx}
\usepackage{subfigure}
\usepackage{color}
\usepackage{algorithm}
\usepackage[noend]{algorithmic}
\usepackage{float}

\newtheorem{theorem}{Theorem}
\newtheorem{definition}{Definition}

\newtheorem{proposition}{Proposition}

\usepackage{geometry}
\begin{document}

\title{Scalable Steiner Tree for Multicast Communications in Software-Defined Networking}


\author{\IEEEauthorblockN{Liang-Hao Huang,~Hui-Ju Hung,~Chih-Chung Lin, and~De-Nian Yang}
\IEEEauthorblockA{Academia Sinica, Taipei, Taiwan\\
\{lhhuang, hjhung, chchlin, dnyang\}@iis.sinica.edu.tw}}

\maketitle

\begin{abstract}
Software-Defined Networking (SDN) enables flexible network resource
allocations for traffic engineering, but at the same time the scalability
problem becomes more serious since traffic is more difficult to be
aggregated. Those crucial issues in SDN have been studied for unicast but have not
been explored for multicast traffic, and addressing those issues for
multicast is more challenging since the identities and the number of
members in a multicast group can be arbitrary. In this paper, therefore, we
propose a new multicast tree for SDN, named Branch-aware Steiner Tree (BST).
The BST problem is difficult since it needs to jointly minimize the numbers of the edges and the branch nodes in a tree, and we prove that it is NP-Hard and
inapproximable within $k$, which denotes the number of group members. We further
design an approximation algorithm, called Branch Aware Edge Reduction
Algorithm (BAERA), to solve the problem. Simulation results demonstrate that
the trees obtained by BAERA are more bandwidth-efficient and scalable than
the shortest-path trees and traditional Steiner trees. Most importantly,
BAERA is computation-efficient to be deployed in SDN since it can generate a
tree on massive networks in small time.
\end{abstract}

\begin{IEEEkeywords}
SDN, multicast, NP-Hard, traffic engineering, scalability
\end{IEEEkeywords}

\section{Introduction}

\IEEEPARstart{S}{oftware-Defined} Networking (SDN) is an emerging
architecture that is manageable, dynamic, cost-effective, and adaptable,
making it ideal for the high-bandwidth, huge data, and dynamic nature of
numerous network services  \cite{sdnwebsite}. This novel architecture decouples the network
control and forwarding functions. It enables the network control to become
directly programmable and the underlying infrastructure to be abstracted for
varied applications. The OpenFlow protocol has been recognized as a crucial
element for building SDN solutions \cite{sdnwebsite, McKeown2008, OpenFlow2013}.

SDN comprises two main components: SDN controller (SDN-C) and SDN
forwarding element (SDN-FE) \cite{OpenFlow2013}. Compared with the
traditional shortest-path routing, SDN-C enables the centralized computation on
unicast routing for traffic engineering \cite{Agarwal2013} to improve the network
throughput. Nevertheless, since the routing paths no longer need to be the
shortest ones, the paths can be distributed flexibly inside the network and
thus are more difficult to be aggregated in the flow table of SDN-FE, and
the scalability has been regarded as a serious issue to deploy SDN in a
large network \cite{Agarwal2013, Kanizo2013}.

\begin{figure}[h]
\centering
\subfigure[Original network]{
    \label{fig1:subfig:a}
    \includegraphics[scale=0.50]{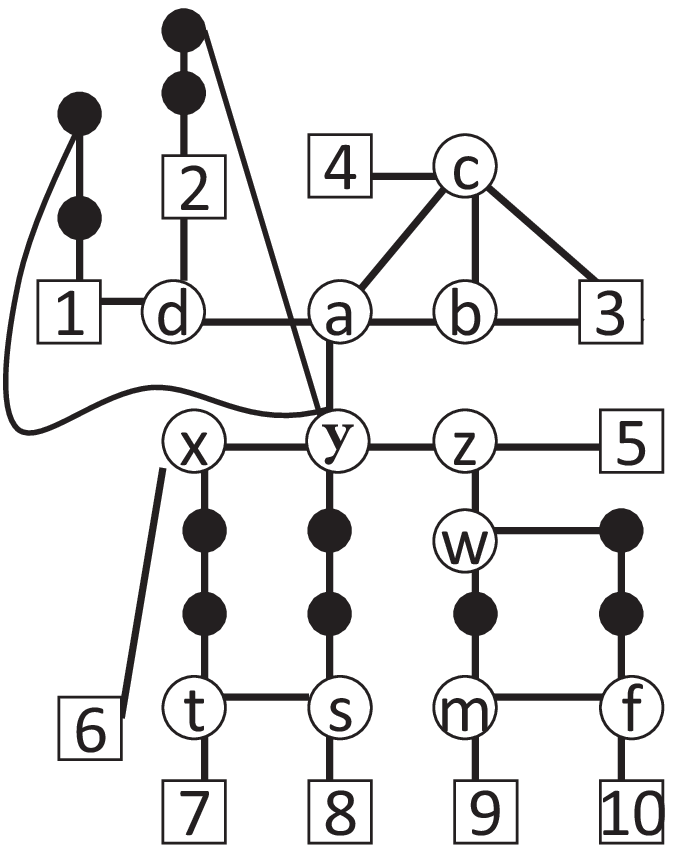}}
\subfigure[Shortest-path tree]{
    \label{fig1:subfig:b}
    \includegraphics[scale=0.50]{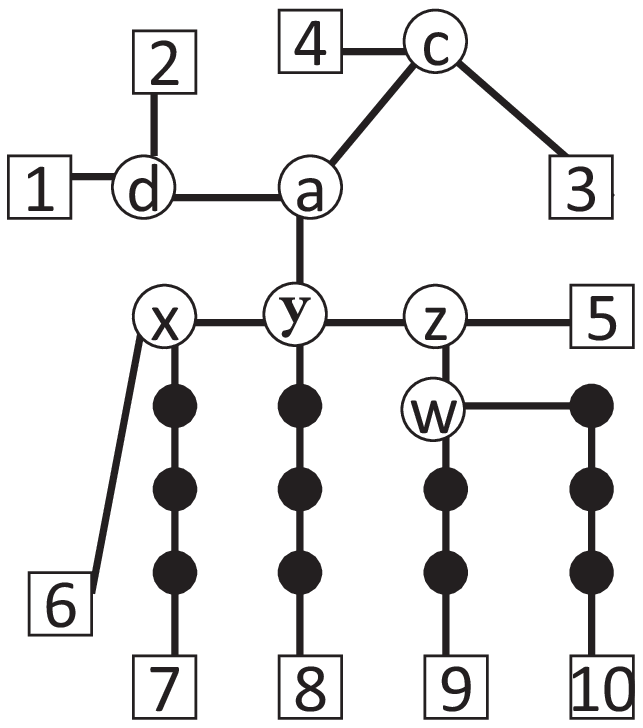}}
\subfigure[Steiner tree]{
    \label{fig1:subfig:c}
    \includegraphics[scale=0.50]{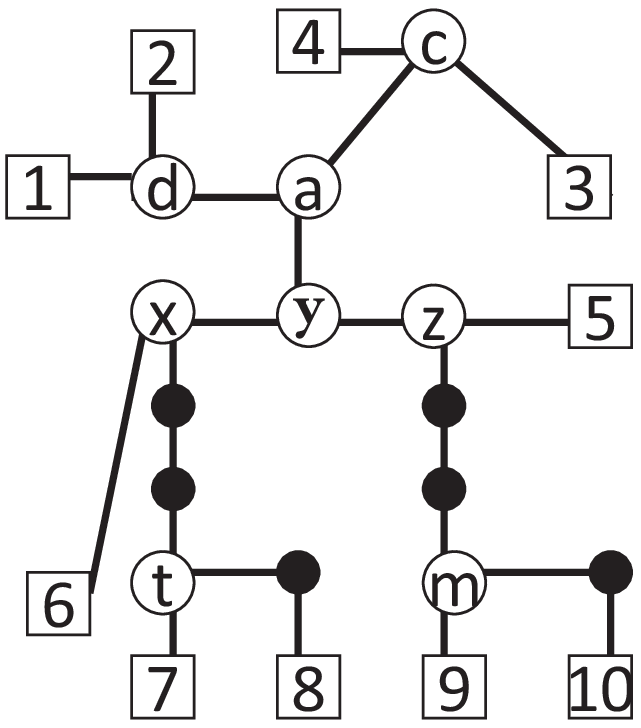}}
\subfigure[Branch-aware Steiner tree]{
    \label{fig1:subfig:d}
    \includegraphics[scale=0.50]{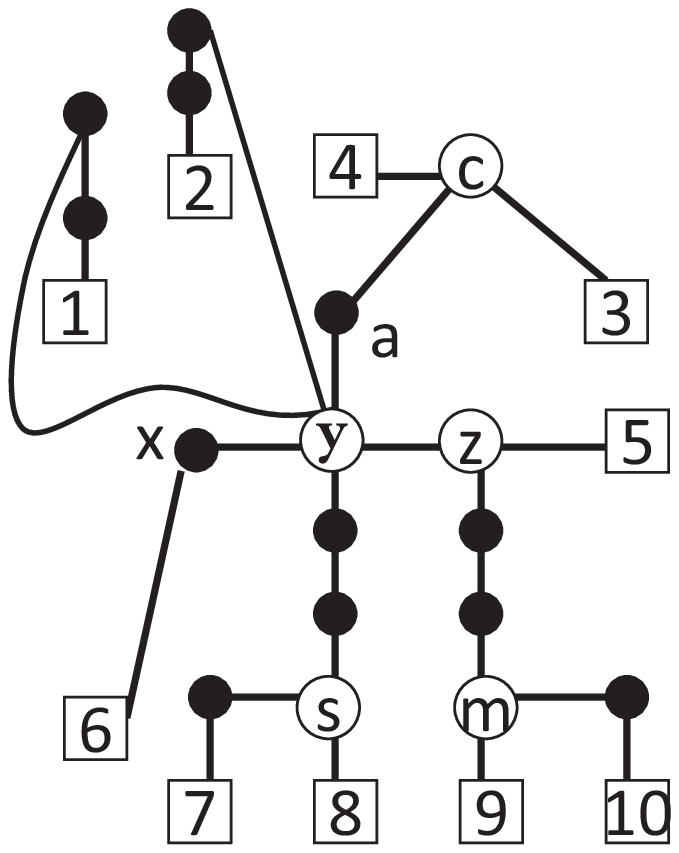}}
\caption{An example of multicase tree}
\label{fig1:subfig}
\end{figure}

Multicast is an efficient technique for point-to-multipoint (P2M) and
multipoint-to-multipoint (M2M) communications because it exploits a tree,
instead of disjoint paths, in the routing of the traffic. Current multicast
standard on Internet, i.e., PIM-SM \cite{Fenner2006}, employs a shortest-path tree to
connect the terminal nodes in a multicast group, where a terminal node is a designated router connecting to a LAN with at least one user client joining the group \cite{Cain2002}. Traffic engineering is
difficult to be supported in a shortest-path tree since the path from the
root, i.e., the traffic source in P2M or the rendezvous point in M2M in
PIM-SM, to each destination in the tree is still the shortest path. By contrast, a
Steiner Tree (ST) \cite{Takahashi1980} in Graph Theory is more promising because it minimizes
the network resource consumption, i.e., the number of edges in a tree,
required for a multicast group. However, finding an ST is more computation
intensive and thus is difficult to be deployed as a distributed protocol on
Internet. By contrast, now it becomes feasible by first finding an ST in
SDN-C and then storing the forwarding information in the group tables of
SDN-FEs on the tree.

Similar to unicast traffic engineering in SDN, multicast traffic engineering
also suffers from the scalability problem since each SDN-FE in the tree
needs to store a forwarding entry in the group table for each multicast
group. Nevertheless, the scalability problem for multicast communications is
even more serious since the number of possible multicast group is $O(2^{n})$%
, where $n$ is the number of nodes in a network, and the number of possible
unicast connections is $O(n^{2})$. To remedy this issue, a promising way is
to exploit the \textit{branch forwarding technique} \cite{Yang2008, YangLiao2008, Tian1998, Stoica2000, Wong2000}, which stores the
entries in only the branch nodes, instead of every node, of a multicast
tree. More specifically, a branch node in a tree is the node with at least
three incident edges, such as white circle nodes in Fig. 1, and the square nodes are the terminal nodes. To
minimize the total number of edges in an ST, the path connecting two
neighboring branch nodes (such as nodes \textit{c} and \textit{y} in Fig. 1(d))
needs to be the shortest path between them. Note that an ST is not a
shortest-path tree because the branch nodes can be located anywhere in the
network. This branch forwarding technique can remedy the multicast
scalability problem since packets are forwarded in a unicast tunnel from the
logic port of a branch node in SDN-FE \cite{OpenFlow2013}. In other words, all nodes
in the path (such as black circle nodes in
Fig. 1) exploit unicast forwarding in the tunnel and are no longer necessary to maintain a forwarding entry for the multicast
group.

To effectively address the multicast scalability problem in SDN, it is
crucial to minimize the number of branch nodes in a tree. However, this
important factor has not been considered in ST. In this paper, therefore, we
propose a new multicast tree for SDN, named \textit{Branch-aware Steiner Tree%
} (\textit{BST}). The objective of BST problem is to minimize the summation
of the number of edges and the number of branch nodes in the tree, where a
branch node can be assigned a higher weight to further improve the
scalability. Fig. 1 presents an illustrative example with the weight of each
branch node set as 20. Square nodes are the terminal nodes that are required
to be connected in a tree, while the black and white circle nodes are the
other nodes in the network. Fig. 1(a) is the network topology. The
shortest-path tree in Fig. 1(b) includes 27 edges and 7 branch nodes with
the total cost of the tree as $27+7\times 20=167$. The Steiner tree in Fig.
1(c) has 23 edges and 8 branch nodes with the total cost as $23+8\times
20=183$. By contrast, Fig. 1(d) presents the BST with 26 edges and 5 branch
nodes and the total cost as $26+5\times 20=126$. Therefore, compared with the
shortest-path trees on Internet, BST effectively reduces the network
resource consumption by minimizing the number of edges in the tree. Compared
with ST, more BSTs can be supported in SDN since the number of branch nodes
is effectively minimized.

Finding an BST is very challenging. The ST problem is NP-Hard but can be
approximated within ratio 1.55 \cite{Robins2000} and is thus in APX of complexity theory. 
In other words, there exists an approximation algorithm for ST that can find
a tree with the total cost at most 1.55 times of the optimal solution. By
contrast, we prove that BST is NP-Hard but cannot be approximated within
\textit{k}, which denotes the number of terminal nodes in a multicast group. In
other words, the BST problem is more difficult to be approximated. To
effectively solve BST, we propose a \textit{k}-approximation algorithm,
named \textit{Branch Aware Edge Reduction Algorithm} (BAERA), that can be
deployed in SDN-C. BAERA includes
two phases, Edge Optimization Phase and Branch Optimization Phase, to
effectively minimize the number of edges and branch nodes. Since no $%
(k^{1-\epsilon })$-approximation algorithm exists in BST for arbitrarily
small $\epsilon >0$, BAERA achieves the best approximation ratio.

{The rest of this paper is organized as follows. Section II briefly
summarizes the literature on SDN traffic engineering, SDN flow table
scalability, multicast scalability, and the Steiner tree. Section III formally
presents the problem formulation with Integer Programming and the hardness
result. We design a \textit{k}-approximation algorithm in Section IV, and
Section V presents the simulation results to evaluate the performance of the
proposed algorithm in real networks. We conclude this paper in Section VI.}

\section{Related Works}

Previous works have extensively explored the issues on traffic engineering
and flow table scalability for \textit{unicast traffic} in SDN. Mckeown et al. \cite{McKeown2008} pointed out that OpenFlow can
be deployed with heterogeneous switches. Sushant et al. \cite{Sushant2013} shared their experience of SDN development for the private WAN of Google Inc.
Qazi et al. \cite{Qazi2013} proposed a new system design using SDN for the middleboxes (e.g., firewalls, VPN gateways, proxies).
Agarwal et al. \cite{Agarwal2013} considered the incremental deployment of traffic engineering in the case
where a SDN-C controls only a few SDN-FEs in the network, and the rest of
the network adopts a standard routing protocol, such as OSPF. The merits of
traffic engineering brought by only a limited number of SDN-capable nodes
are demonstrated. Mueller et al. \cite{Mueller2013} presented a cross-layer framework
in SDN, which integrates a novel dynamic traffic engineering approach with
an adaptive network management, to bridge the gap between the network and
application layers for overall system optimizations.

On the other hand, flow table scalability is crucial to enable a large-scale
deployment of SDN. For unicast traffic, Kanizo et al. \cite{Kanizo2013} pointed out that
the restriction on table sizes is the major bottleneck in SDN and proposed a
framework, called Palette, to decompose a large SDN table into small ones
and then distribute them across the network. Lee et al. \cite{Lee2013} observed that
Data Center traffic frequently meets few elephant flows and a lot of mice
flows. However, elephant flows are inclined to be evicted because of the
limited flow table sizes. They proposed a differential flow cache framework
that uses a hash-based cache placement and localized Least Recently Used
(LRU)-based replacement to reduce the loss of elephant flows.

The scalability issue is more serious in multicast, and the previous works
\cite{Yang2008}, \cite{YangLiao2008, Tian1998, Stoica2000, Wong2000} have demonstrated that the branch forwarding technique is a promising
way since forwarding from a branch node to a neighbor branch node or terminal node can exploit the
existing unicast tunneling technique, and tunneling can be facilitated in
SDN with logic ports specified in the group table \cite{OpenFlow2013}. In other words, the
intermediate nodes between two neighbor branch routers no longer need to
store a multicast forwarding entry for the tree. However, the above works
were designed for shortest-path trees and did not explore the possibility of more flexible multicast routing. On the other hand, Steiner tree \cite{Takahashi1980} can
effectively minimize the bandwidth consumption in a network, but so far it
is not adopted on Internet since finding the optimal Steiner tree is more
computation intensive and thus difficult to be deployed as a distributed
protocol. To remedy this issue, overlay
Steiner trees \cite{YangL07IPDS, Aharoni1998} for P2P environments are proposed, where only the terminal nodes can
act as branch nodes. Nevertheless, the merit of traffic engineering from the
above work is limited since no other router can act as the branch node to reduce the bandwidth consumption. Moreover,
multicast scalability is not studied in the above works. Therefore, the
above works are difficult for bandwidth-efficient and scalable multicast in
SDN.


\section{Preliminaries}

\subsection{Problem Formulation}

In this paper, we propose a scalable and bandwidth-efficient multicast tree
for SDN, called \textit{Branch-aware Steiner Tree (BST)}. This paper aims to
minimize the bandwidth consumption (i.e., the total number of links/edges)
and the number of forwarding entries maintained for the multicast group
(i.e., the total branch nodes). Therefore, the BST problem is to find a tree
connecting a given set of terminal nodes such that the sum of the number of
edges and the number of branch nodes is minimized, where a branch node can
be assigned a larger weight $w$ to ensure a higher scalability.\footnote{%
Note that this problem can be simply extended to support different weights
on each edges and each nodes. For example, a congested edge or a node with
the group table almost fulled can be assigned a higher weight.}

\begin{definition}
Consider a network $G(V,E)$, where $V$ and $E$ denote the set of nodes and
edges, respectively. Given $G(V,E)$, a terminal node set $K\subseteq V$, and
a non-negative value $w$, the BST problem is to find a tree $T$ spanning the
terminal node set $K$ such that $c(T)+b(T)w$ is minimized, where $c(T)$ is
the number of edges on $T$, and $b(T)$ is
the number of branch nodes (i.e., nodes with the degree at least 3 on $T$).
\end{definition}

In BST, a network operator can increase the scalability of multicast in SDN by
assigning a larger weight $w$ for branch nodes. Compared with ST, $c(T)$ may
slightly increase, but much fewer branch nodes will be selected in $T$.
Compared with the shortest-path trees adopted on Internet currently, BST
allows more flexible routing of a tree and thus can effectively reduce the
network resource consumption and improve the scalability in SDN.

In the following, we first formulate the BST\ problem as an Integer
Programming problem. Afterwards, we show that the BST problem is very
challenging in complexity theory by proving that it is NP-Hard and not able
to be approximated within $k^{c}$ for every $c<1$.

\subsection{Integer Programming}
\label{subsec:integer_programming}

Let $N_{v}$ denote the set of neighbor nodes of $v$ in $G$, and $u$ is in $%
N_{v}$ if $e_{u,v}$ is an edge from $u$ to $v$ in $E$. Let any terminal node $r$ act as the
root of $T$, i.e., the source, and the destination set $L$ contains the other terminals in $K$, i.e., $%
L=K-\left \{ r\right \} $. The output tree $T$ needs to ensure that there is
only one path in $T$ from $r$ to every node in $L$. To achieve this goal,
our problem includes the following binary decision variables. Let binary
variable $\pi _{l,u,v}$ denote if edge $e_{u,v}$ is in the path from $%
r $ to a destination node $l$ in $L$. Let binary variable $\varepsilon _{u,v}$
denote if edge $e_{u,v}$ is in $T$, where $\varepsilon _{u,v}=\varepsilon
_{v,u}$. Let binary variable $\beta _{v}$ denote if $v$ is a branch node in $%
T$. Intuitively, when we are able to find the path from $r$ to each destination
node $l$ with $\pi _{l,u,v}=1$ on every edge $e_{u,v}$ in the path, the
routing of the tree with $\varepsilon _{u,v}=1$ for every edge $e_{u,v}$ in $%
T$ can be constructed with the union of the paths from $r$ to all destination
nodes in $L$, and every branch node $v$ in $T$ with $\beta _{v}=1$ in $T$
can be identified accordingly.

Most importantly, to guarantee that the union of the paths is a tree, i.e.,
a subgraph without any cycle, the objective function of our Integer
Programming formulation (IP) is as follows.%
\[
\min \sum \limits_{e_{u,v}\in E}\varepsilon _{u,v}+\sum \limits_{v\in
V}w\times \beta _{v}.
\]%
If the tree $T$ contains any cycle, $T$ is not optimal since we are able to
remove at least one edge from the cycle to reduce the objective value, and
ensure that there still exist a path from $r$ to every destination node $l$ in $L$.
To find $\varepsilon _{u,v}$ and $\beta _{v}$ from $\pi _{l,u,v}$, our IP
formulation includes the following constraints.

\begin{center}
\begin{tabular}{cc}
$\sum \limits_{v\in N_{r}}\pi _{l,r,v}-\sum \limits_{v\in N_{r}}\pi
_{l,v,r}=1$, $\forall l\in L,$ & (1)\\
\vspace{2pt}
$\sum \limits_{u\in N_{l}}\pi _{l,u,l}-\sum \limits_{u\in N_{l}}\pi
_{l,l,u}=1$, $\forall l\in L,$ & (2)\\
\vspace{2pt}
$\sum \limits_{v\in N_{u}}\pi _{l,v,u}=\sum \limits_{v\in N_{u}}\pi _{l,u,v}$%
, \\ $\forall l\in L$, $\forall u\in V,u\neq l,u\neq r,$ & (3)\\
\vspace{2pt}
$\pi _{l,u,v}\leq \varepsilon _{u,v}$, $\forall l\in L$, $\forall e_{u,v}\in
E,$ & (4)\\
\vspace{2pt}
$\frac{1}{\left \vert N_{u}\right \vert }\left( -2+\sum \limits_{v\in
N_{u}}\varepsilon _{u,v}\right) \leq \beta _{u}$, $\forall u\in V.$ & (5)%
\end{tabular}
\end{center}

The first three constraints, i.e., (1), (2), and (3), are the flow-continuity constraints to find the
path from $r$ to every destination node $l$ in $L$. More specifically, $r$ is the
flow source, i.e., the source of the path to every destination node $l$, and
constraint (1) states that the net outgoing flow from $r$ is one, implying
that at least one edge $e_{r,v}$ from $r$ to any neighbor node $v$ needs to
be selected with $\pi _{l,r,v}=1$. Note that here decision variables $\pi
_{l,r,v}$ and $\pi _{l,v,r}$ are two different variables because the flow is
directed. On the other hand, every destination node $l$ is the flow destination,
and constraint (2) ensures that the net incoming flow to $l$ is one,
implying that at least one edge $e_{u,l}$ from any neighbor node $u$ to $l$
must be selected with $\pi _{l,u,l}=1$. For every other node $u$, constraint
(3) guarantees that $u$ is either located in the path or not. If $u$ is
located in the path, both the incoming flow and outgoing flow for $u$ are at
least one, indicating that at least one binary variable $\pi _{l,v,u}$ is $1$
for the incoming flow, and at least one binary variable $\pi _{l,u,v}$ is $1$
for the outgoing flow. Otherwise, both $\pi _{l,v,u}$ and $\pi _{l,u,v}$ are
$0$. Note that the objective function will ensure that $\pi _{l,v,u}=1$ for
at most one neighbor node $v$ to achieve the minimum cost. In other words,
both the incoming flow and outgoing flow among $u$ and $v$ cannot exceed $1$.

Constraints (4) and (5) are formulated to find the routing of the tree and
its corresponding branch nodes, i.e., $\varepsilon _{u,v}$ and $\beta _{v}$.
Constraint (4) states that $\varepsilon _{u,v}$ must be $1$ if edge $e_{u,v}$ is
included in the path from $r$ to at least one $l$, i.e., $\pi _{l,u,v}=1$. The tree $T$ is the
union of the paths from $r$ to all destination nodes. Note that here $%
\varepsilon _{u,v}$ and $\varepsilon _{v,u}$ represent the same binary
decision variable because $T$ is not directed. In other words, $\varepsilon
_{u,v}=1$ if edge $e_{u,v}$ is in a path (i.e., a directed flow) from either
direction. The last constraint is the most crucial one. For each node $u$,
if the degree of $u$ is at least $3$ in $T$, $\sum_{v\in
N_{u}}\varepsilon _{u,v}\geq3$ holds, and thus the left-hand-side of
constraint (5) becomes positive, thereby enforcing that $\beta _{u}=1$ and $%
u $ acts as a branch node. Otherwise, the left-hand-side of constraint (5)
is $0$ or negative, allowing $\beta _{u}$ to be $0$ to minimize the cost in
the objective function. In this case, node $u$ is not a branch node in $T$.

\subsection{Hardness result}

The BST problem is NP-Hard because it is equivalent to the ST\ problem when $%
w$ is 0. In other words, the ST problem is a special case of the BST
problem. However, the BST is much more challenging because the ST problem
can be approximated within ratio 1.55 \cite{Robins2000} and is thus in APX in
complexity theory, 
but we find out that BST is much more difficult to be
approximated. The following theorem proves that the BST problem cannot be
approximated within $k^{c}$ for every $c<1$, by a gap-introducing reduction
from the Hamiltonian path problem, which determines whether there exists a
path going through every node on a graph exactly once.


\begin{theorem}
\label{hardness} For any $\epsilon >0$, there exists no $k^{1-\epsilon }$
approximation algorithm for the BST problem, assuming P $\neq $ NP.
\end{theorem}

\begin{proof}
We prove the theorem with the gap-introducing reduction from the Hamiltonian
path problem. For an instance $G_{H}(V_{H},E_{H})$ of the Hamiltonian path
problem with any node $v$ on $G_{H}$, we build an instance of the BST
problem on $G(V,E)$, such that \newline
$\bullet $ if a Hamiltonian path starting at $v$ exists in $G_{H}$, $\mathrm{%
OPT}(G)\leq 2h$, and\newline
$\bullet $ if no Hamiltonian path starting at $v$ exists in $G_{H}$, $%
\mathrm{OPT}(G)>2hk^{1-\epsilon }$ \newline
, where $h$ is the number of nodes in $G$ and $\mathrm{OPT}(G)$ is the
optimal solution of $G$ for the BST problem.

We first detail how to build the instance of the BST problem from the
Hamiltonian path problem. For any given $G_{H}$, we construct a new graph $G$
which consists of $n^{p}$ copies of $G_{H}$, where $n$ is the number of
nodes in $G_{H}$ and $p$ is the smallest integer following $p\geq \frac{2}{%
\epsilon }$. One additional node $x$ is added to $G$ to connect to the node $%
v$ of each of the $n^{p}$ copies.
The $K$ is set to $V-\{x\}$ and $w$ is set to $h$, where $h$ is the number
of nodes in $G$, i.e., $h=(n^{p})\times n+1$.

If $G_{H}$ has a Hamiltonian path starting at $v$, consider a tree rooted at $x$, which
includes 1) the edges between $x$ and $v$ of all copies and 2) the edges on
the Hamiltonian path of all copies. The tree is a feasible solution of the
BST problem with only one branch node $x$, and it can act as an upper bound
of the BST in $G$.
Thus, $\mathrm{OPT}(G)\leq h+(h-1)<2h$. On the other hand, if $G_{H}$ does
not have a Hamiltonian path starting at $v$, there must exist at least one
additional branch node in each copy of $G$. Hence, $\mathrm{OPT}%
(G)>hn^{p}\geq 2hn^{p-1}=2h(n^{p+1})^{\frac{p-1}{p+1}}\geq
2h(n^{p+1})^{1-\epsilon }=2hk^{1-\epsilon }$. Since $\epsilon $ can be
arbitrarily small, for any $\epsilon >0$, there is no $k^{1-\epsilon }$
approximation algorithm for the BST problem, assuming P $\neq $ NP. The theorem
follows.
\end{proof}

\section{Algorithm Design}

For BST, the shortest-path tree is not a good solution since the shortest
path for each node $v$ in $K$ is constructed individually. With the aim to
minimize the number of the edges, substituting the shortest path of $v$ with a
longer path can reduce the total edge number when the path mostly overlaps
with the path to another node $v^{\prime }\in K$ \cite{Takahashi1980}. Therefore, it is expected
that aggregating two paths that share more common edges can effectively
reduce the number of edges in $T$. Nevertheless, aggregating two paths that
partially overlap will generate a new branch node, and more branch nodes are
inclined to be created when more paths are aggregated. Without considering
the number of branch nodes created, the solution quality may deteriorate
even though the number of edges in $T$ is effectively reduced. In the
following, therefore, we propose a $k$-approximation algorithm for BST,
called \textit{Branch Aware Edge Reduction Algorithm} (BAERA), to jointly
minimize the numbers of edges and branch nodes in $T$. As Theorem~\ref%
{hardness} proves that no $(k^{1-\epsilon })$-approximation algorithm for
any $\epsilon >0$ for the BST problem, BAERA achieves the best approximation
ratio. Due to space constraint, the pseudo code is presented in \cite{DBLP}.

BAERA includes two phases: 1) Edge Optimization Phase and 2) Branch
Optimization Phase. In the first phase, BAERA iteratively chooses and adds a
terminal node in $K$ to the solution tree $T(V_{T},E_{T})$ for constructing
a basic BST, where $V_{T}$ and $E_{T}$ denote the nodes and edges currently
in $T$, respectively at each itereation. Initially, a random root node is added to $V_{T}$.
Afterwards, for each terminal node $v\in K$ that is not in $V_{T}$, BAERA
first finds the minimal distance $d_{v,T}$ from $v$ to $T$. Precisely, let $%
p_{v,u}$ denote the shortest path from $v$ to $u$ on the network $G$, and $%
\left \vert p_{v,u}\right \vert $ is the number of edges in $p_{v,u}$. The
minimal distance $d_{v,T}$ from $v$ to $T$ is $\min_{u\in V_{T}}\left \vert
p_{v,u}\right \vert $, and $u$ here represents the node closest to $v$ in $T$%
. After finding $d_{v,T}$ for every $v$, BAERA extracts the node $v_{\min }$
with the smallest $d_{v,T}$, i.e., $v_{\min }=\arg \min_{v\in
K-V_{T}}d_{v,T} $ and adds $p_{v,u}$ to $T$.\footnote{%
In this paper, we connect $v_{\min }$ to $T$ via the shortest path.
Nevertheless, it is also allowed to connect $v_{\min }$ to $T$ with an
alternate path derived according to unicast traffic engineering \cite%
{Agarwal2013} to meet the unicast traffic requirements.} Most importantly,
to avoid constantly generating a new branch node, BAERA will choose $p_{v_{\min },u}$,
i.e., let the node $v_{\min }$ connect to $u$ which already acted as a branch node in $%
T$, if there are multiple $v_{\min }$ sharing the same minimal distance $%
d_{v_{\min },T}$. Edge Optimization Phase ends when all nodes in $K$ are
added to $V_{T}$.

Fig. 2(a) presents an example of Edge Optimization Phase, where node 1 is
the root. Node 2 is first connected to node 1 with 2 edges via node $%
d $. Node 3 is then connected to $d$ with 3 edges via nodes $b$ and $a$%
. Node 4 is then connected to $b$ with 2 edges via $c$. Afterwards,
node 5 and node 6 are connected to $T$ sequentially. For node 7
and node 8, note that $d_{7,T}$ and $d_{8,T}$ are both 4 in Fig. 1(a), and considering $%
p_{8,y}$ will not generate another branch node, therefore node 8
is first connected to $T$ and then node 7 is connected to $T$ via the created branch node $s$.
Afterwards, node 9 and node 10 are connected to $T$ sequentially.

\begin{figure}[t]
\centering
\subfigure[Edge Optimization Phase]{
    \label{fig2:subfig:a}
    \includegraphics[scale=0.50]{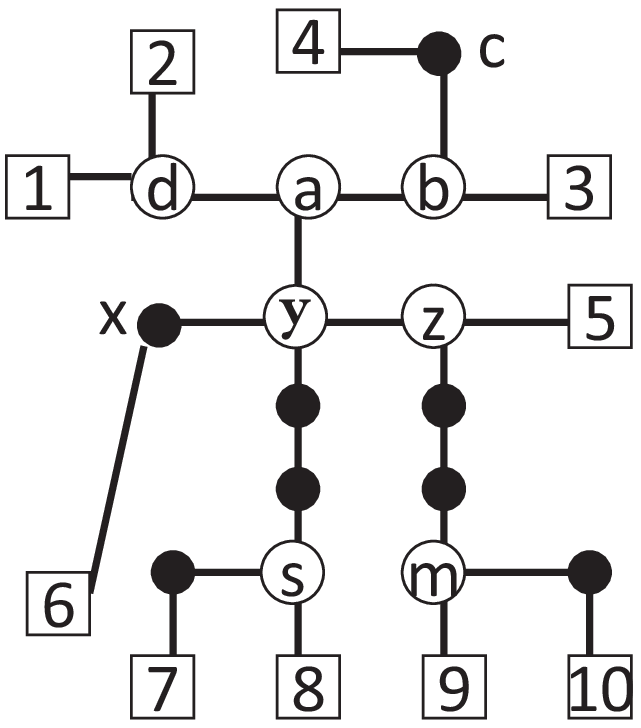}}
\smallskip
\subfigure[Deletion Step]{
    \label{fig2:subfig:b}
    \includegraphics[scale=0.50]{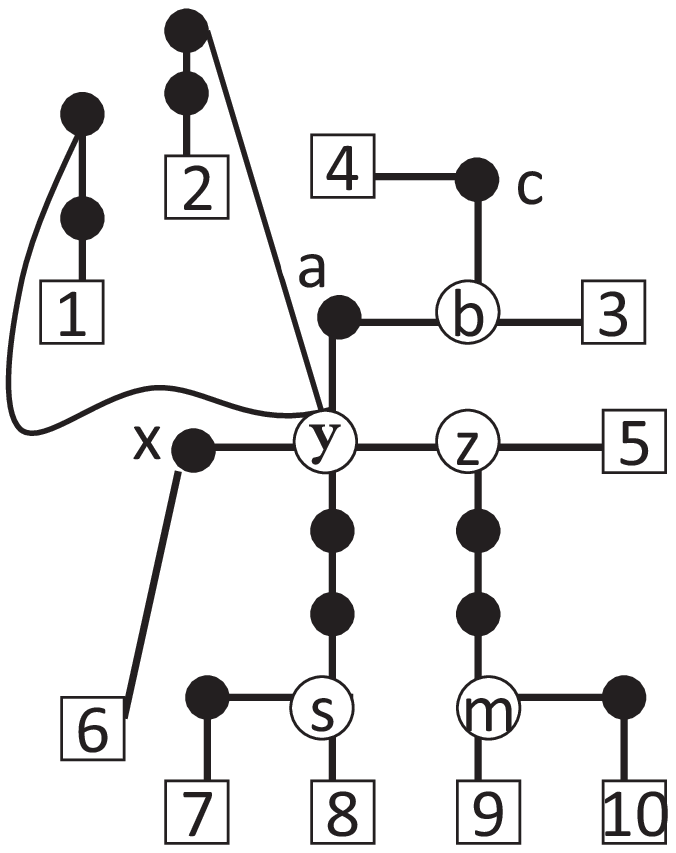}}
\caption{An example of BAERA (refer to Fig.1)}
\label{fig2:subfig}
\end{figure}

Afterwards, Branch Optimization Phase re-routes the tree $T$ to reduce the
number of branch nodes. Intuitively, if more branch nodes are allowed in $T$%
, the nodes in $K$ can connect to $T$ with shorter paths, as
the plan in Edge Optimization Phase. Nevertheless, as the weight $w$ of a
branch node increases, it is necessary for a terminal node to pursue a
longer path that directly connects to an existing branch node in $T$ to
avoid creating a new branch node. To address this issue, Branch Optimization
Phase includes two steps: 1) Deletion Step and 2) Alternation Step. Deletion
Step first tries to remove some branch nodes in $T$ obtained from Edge
Optimization Phase, and then Alternation Step tries to iteratively move each of remaining branch nodes to its neighbor node. In the above two steps, the
solution $T$ will be replaced by the new one only if its objective value $%
c(T)+b(T)w$ is improved (i.e., reduced).

More specifically, Deletion Step first sorts the branch nodes by the ascending order of the
degree in $T$. In other words, a branch node owning fewer neighbor branch
nodes and neighbor terminal nodes\footnote{%
Herein, the examples for the neighbor terminal node and neighbor branch node are presented. In Fig. 2(b), node 2 is a neighbor terminal node of $y$ because there is
no other branch node or terminal node between them, while node 4 is not the neighbor terminal node of $y$. Node $b$ is a neighbor branch node of $y$, but is not a neighbor branch node of $s$.}
will be examined first because the solution has a higher chance to be
improved.
When a branch node $v_{d}$ is removed, because $T$ is partitioned into multiple connected components, $v_{d}$'s neighbor branch node and neighbor terminal node will correspond to different
connected components.
Deletion Step will re-route $v$ to the $v$'s closest branch node $u$ in another connected component via its shortest path $p_{v,u}$ to merge the two
connected components.\footnote{%
If a cycle is created by adding $p_{v,u}$, the longest path between two
neighbor branch nodes in the cycle can be removed.} This process is repeated
such that different connected components will be connected together to
create a new tree.
Fig. 2(b) presents an example of deleting branch node $d$ from Fig. 2(a).
After $d$ is deleted, node 1 and node 2 are re-routed to the other connected component's node $a$ via node $y$. Therefore, the number of
branch nodes can be reduced when Deletion Step ends.

Afterwards, Alternation Step sorts the branch nodes in the ascending order
of the degree again. This step tries to move each branch node $v_{a}$ to a
neighbor node $v_{n}$. For each neighbor branch node or neighbor terminal
node $v$ of $v_{a}$, $p_{v,v_{a}}$ is replaced by $p_{v,v_{n}}$.\footnote{%
Any cycle created by adding $p_{v,v_{n}}$ is also necessary to be removed.}
This step will choose the neighbor node $v_{n}$ leading to the most
reduction on the objective value $c(T)+b(T)w$, and each branch nodes can be
moved multiple times until no neighbor node is able to reduce the objective
value. The difference between Alternation Step and Deletion Step is that
here every $v$ in different connected component will connect to the same
node (i.e., $v_{n}$), leading to a chance on the reduction of the edge
number. {Fig. 1(d) presents the result of altering branch node }$b$ to its
neighbor $c$ in Fig. 2(b). Paths $p_{y,b}$ and $p_{3,b}$ are replaced by
paths $p_{y,c}$ and $p_{3,c}$ with $c(T)$ reduced by $1$.

In the following, we prove that BAERA with the above two phases is a $k$%
-approximation algorithm if the optimal solution includes at least one
branch node. On the other hands, when the optimal solution has no branch
node, it will become a path, instead of a tree.  We will discuss this case later.

\begin{theorem}
BAERA is a $k$-approximation algorithm for the BST problem.
\label{thm:k_ratio}
\end{theorem}

\begin{proof}
In Edge Optimization Phase, since $T$ is constructed by adding shortest
paths to $T$, $c(T)=\sum_{v\in K}d_{v,T}$ as explained early in this
section. Because $d_{v,T}=\min_{u\in V_{T}}\left \vert p_{v,u}\right \vert $
and the root node $r\in V_{T}$, $d_{v,T}\leq d_{v,r}$, where $d_{v,r}$ is
the number of edges in the shortest path from $v$ to $r$. Let $T^{\ast }$
denote the optimal BST, and $d_{v,r}^{\ast }$ denote number of edges in the
path from $v$ to $r$ on $T^{\ast }$, which may not be the shortest path
between $v$ and $r$ in $G$. In other words, $d_{v,r}\leq d_{v,r}^{\ast }$.
Apparently, $d_{v,r}^{\ast }\leq c(T^{\ast })$, and thus we conclude that $%
c(T)=\sum_{v\in K}d_{v,T}\leq \sum_{v\in K}d_{v,r}\leq \sum_{v\in
K}d_{v,r}^{\ast }\leq k\ast c(T^{\ast })$ after the first phase ends. On the
other hand, $T$ cannot have more than $k$ branch nodes because each step in
this phase creates at most one branch node. Therefore, $b(T)\leq k\ast
b(T^{\ast })$ since $b(T^{\ast })\geq 1$, and the tree $T$ generated in the
first phase is $k$-approximated. Since the second phase re-routes the tree
only if the objective value $c(T)+b(T)w$ can be reduced, the tree $T$
outputed in the second phase is also $k$-approximated. The theorem follows.
\end{proof}

In the following, we discuss the cases when  the optimal solution has no branch node, i.e., the optimal solution is a path, instead of a tree. Let $P^{\ast }$ denote the optimal BST.

\begin{proposition}
If $w\leq k$, then BAERA is a $2k$-approximation algorithm.
\label{thm:2k_approx}
\end{proposition}

\begin{proof}
Denote $T$ as the tree generated by BAERA. First, since $T$ is constructed by adding shortest paths to $T$, we know that $c(T) \leq k \times c(P^*)$ according to Theorem~\ref{thm:k_ratio}. Second, since BAERA includes at most one additional branch node in each iteration, there are at most $k-2$ branch nodes in $T$, i.e., $k-2 \leq b(T)$.  In addition, since $P^*$ connects all terminals in $K$, the number of edges in $P^*$ must be at least $k - 1$, i.e., $k-1 \leq c(P^*)$. Thus, we obtain $b(T)w \leq (k-2)w \leq c(P^*) \times k$. Therefore, $c(T)+b(T)w \leq 2k \times c(P^*)$ and BAERA is a $2k$-approximation algorithm when $w \leq k$. The theorem follows.
\end{proof}

%

BAERA tends to generate a solution with branch nodes. For the case with a large w, we explore another direction that leverages the Hamiltonian path to find a solution with the performance guarantee. In the following, we first introduce the Ore's Theorem~\cite{Ore1960}, and then prove that an Hamiltonian path must exist if the degree of selected nodes are large enough in  Proposition~\ref{thm:no_branch}.

\begin{theorem}[Ore's Theorem]
Let $G=(V,E)$ be a connected simple graph with $n\geq 3$ vertices.
If for each pair of non-adjacent vertices
$u,v\in V$ such that $\mathrm{deg}(u)+\mathrm{deg}(v)\geq n$, $G$ contains a Hamiltonian cycle.
\end{theorem}

%

\begin{proposition}
Assume there exists a connected subgraph $H$ of $G$ and for each pair of non-adjacent vertices $u, v\in V(H)$, $\mathrm{deg}_{H}(u) + \mathrm{deg}_{H}(v)\geq |V(H)|$.
If $K\subseteq V(H)$ and $|V(H)|\leq (k-1)k$, then we can find a Hamiltonian
path $P$ with $\frac{c(P)}{c(P^{\star })}\leq k$.
\label{thm:no_branch}
\end{proposition}

\begin{proof}
In the following, we discuss the case when $|V(H)|\geq 3$.\footnote{Note that a connected $H$ with $|V(H)|=2$ contains two nodes and a link between them, and the Hamiltonian path can be easily derived.} Since $H$ is connected and for each pair of non-adjacent vertices $u, v\in V(H)$, $\mathrm{deg}_{H}(u) + \mathrm{deg}_{H}(v) \geq |V(H)|$ holds, $H$ has a Hamiltonian cycle $C$ according to Ore's Theorem. Thus we can find a path $P$ from $C$ such that the start node and end node are in $K$ and connects all terminal nodes which satisfies $c(P) \leq |V(H)|\leq (k-1)k$. Since $c(P^{*})\geq (k-1)$ (as mentioned in Proposition~\ref{thm:2k_approx}), we obtain $\frac{c(P)}{c(P^{*})}\leq \frac{(k-1)k}{k-1}=k$. The theorem follows.
\end{proof}

Note that those $H$ can be obtained by examining dense subgraphs \cite{Khuller09ICALP} or $k$-cores \cite{Seidman1983}. Moreover, if such $H$ exists, the corresponding Hamiltonian Path $P$ could be derived by an existing algorithm \cite{Palmer1997}.

\textbf{Time Complexity. }We first find the shortest path between any two
nodes in $G$ with Johnson's algorithm in $O(|V||E|+|V|^{2}\log |V|))$ time
as a pre-processing procedure for quickly lookup afterwards. The advantage is
that the preprocessing only needs to be performed once but can be exploited
during the construction of all BSTs afterwards. In each iteration of Edge Optimization Phase, BAERA finds $%
d_{v,T} $ and extracts $v_{\min }$ in $O(k|V|)$ time, and this phase
requires $O(k^{2}|V|)$ time to connect all terminal nodes to $T$.

In Branch Optimization Phase, let $B$ denote the set of branch nodes in $T$.
Let $\delta _{T}$ denote the maximal degree of a node in $T$, and $\delta
_{T}\leq k$ and $\delta _{T}\leq \delta _{G}$ must hold, where $\delta _{G}$
is the maximal degree of a node in $G$. Deletion Step first sorts the branch
nodes in the ascending order of the degree in $T$. Since $|B|\leq k-2$, the
sorting requires $O(k\log k)$ time. We then build a heap for each branch node to store the
shortest-path distance from other branch nodes to $v$ in $O(k\log k)$ time. To
remove a branch node $v_{d}$, it is necessary to connect each neighbor
branch node and neighbor terminal node $v$ to the existing closest branch
node $u$ in $T$ in $O(\log k)$ time. Therefore, Deletion Step takes $%
O(\delta _{T}\log k)$ time to delete a branch node, and thus $O(k\delta _{T}\log k)$ for trying to delete all branch nodes. In Alternation Step, first the branch nodes are sorted in $O(k\log k)$ time.
Then, BAERA tries to move each branch node in order. Note that each branch node $v_{a}$ can be moved at most $O(|V|)$ times, and moving $v_{a}$ to a neighbor takes $O(\delta _{T})$ time. Alternation Step takes $O(k\log k+k\delta _{T}|V|)$ time. Therefore, the time
complexity of Branch Optimization Phase is $O(k\log k+k\delta _{T}|V|)$,
and BAERA takes $O(k^{2}|V|+k\delta _{T}|V|)$ time after the
pre-processing procedure. As shown in Section V later, $\delta _{T}$ is
usually small, and thus the time complexity of BAERA after pre-processing is $O(k^{2}|V|)$. Moreover, $|V|$ in the above analysis represents an upper bound of the cost for scanning the tree $T$. Since the tree size is usually much smaller than $|V|$, the computation cost is actually close to $O(k^2|T|)$.

\begin{algorithm}[t]
\caption{Branch Aware Edge Reduction Algorithm (BAERA)}
\label{alg}
\begin{algorithmic}[1]
\REQUIRE {A network $G=(V,E)$, a nonnegative value $w$ and a terminal set $K$.}
\ENSURE {A Steiner tree $T$.}
\STATE{//Edge Optimization Phase}
\STATE{Choose a terminal node $r$ as the root}
\STATE{$T \leftarrow \{r\}$, $K \leftarrow K-\{r\}$, $A(T)\leftarrow 0$}
\WHILE{$K\neq \emptyset$}
	\FOR{$v\in K$}
		\STATE{$d_{v,T}\leftarrow$ the minimum distance from $v$ to $T$}
		\STATE{$p_{v,T}\leftarrow$ the shortest path from $v$ to $T$}
	\ENDFOR
	\STATE{$S \leftarrow \{x|~ d_{x,T}=\min_{v\in K}d_{v,T}\}$}
	\IF{ there exists a $x\in S$ such that $T\cup p_{x,T}$ does not generate a new branch node}
		\STATE{$T \leftarrow T\cup p_{x,T}$}
	\ELSE{}
		\STATE{Choose a $x\in S$ and $T \leftarrow T\cup p_{x,T}$} 	
	\ENDIF
	\STATE{$K\leftarrow K-\{x\}$}		
\ENDWHILE
\STATE{$A(T)\leftarrow c(T)+b(T)w$ //The weight of the tree $T$}
\STATE{}
\STATE{//Branch Optimization Phase 1) Deletion Step}
\STATE{Obtain an order $\sigma$ which sorts the branch nodes in the ascending order of the degree in $T$}
\FOR{$v_{d}\in \sigma$}
	\STATE{$T'\leftarrow T-\{v_{d}\}$}
	\FOR{neighbor branch node or neighbor terminal node $v$ of $v_{d}$}
		\STATE{Reroute the $v$'s closest branch node $u$ in another connected component via its shortest path $p_{v,u}$}
		\STATE{$T'\leftarrow T'\cup p_{v,u}$}
	\ENDFOR
	\IF{$c(T')+b(T')w < A(T)$}
		\STATE{$T \leftarrow T'$ and $A(T) \leftarrow c(T')+b(T')w$}
	\ENDIF
\ENDFOR
\STATE{}
\STATE{//Branch Optimization Phase 2) Alternation Step}
\STATE{Obtain an order $\sigma$ which sorts the branch nodes in the ascending order of the degree in $T$}
\FOR{$v_{a}\in \sigma$}
	\STATE{$T'\leftarrow T-\{v_{a}\}$}
	\STATE{Choose a neighbor node $v_{n}$ of $v_{a}$}
	\FOR{neighbor branch node or neighbor terminal node $v$ of $v_{a}$}
		\STATE{The shortest path $p_{v,v_{a}}$ is replaced by the shortest path $p_{v,v_{n}}$}
		\STATE{$T'\leftarrow T'\cup p_{v,v_{n}}$}
	\ENDFOR
	\IF{$c(T')+b(T')w < A(T)$}
		\STATE{$T \leftarrow T'$ and $A(T) \leftarrow c(T')+b(T')w$}
	\ENDIF
\ENDFOR
\RETURN{$T$ and $A(T)$}	
\end{algorithmic}
\end{algorithm}

\section{Simulation Results}

In this section, we evaluate BAERA in both real networks and massive
synthetic networks.

\subsection{Simulation Setup}

The simulation is conducted in the following real networks \cite{Zoowebsite}: 1) the Uunet network with 49 nodes and 84 links, and
2) the Deltacom network with 113 nodes and 183 links. Many recent SDN works
\cite{Agarwal2013, Qazi2013} evaluate the proposed approaches in real
networks with at most hundreds of nodes. By contrast, we also evaluate our
algorithm in the networks generated by Inet \cite{Tangmunarunkit2002,
Inetwebsite} with tens of thousands of nodes to test the scalability of
BAERA. In our simulation, $K$ is chosen randomly from $G$.

We compare BAERA with the following algorithms: 1) the shortest-path tree
algorithm (SPT), 2) a Steiner tree (ST) algorithm \cite{Takahashi1980}, and 3)
Integer Programming solver CPLEX \cite{cplexwebsite}, which finds the
optimal solution of the BST problem by solving the Integer Programming
formulation in Section~\ref{subsec:integer_programming}. The performance
metrics include: 1) the objective value of the BST problem $b(T)w+c(T)$, 2)
the number of branch nodes in $T$, 3) the number of edges in $T$, and 4) the
running time. All algorithms are implemented in an HP DL580 server with four
Intel Xeon E7-4870 2.4 GHz CPUs and 128 GB RAM. Each simulation result is
averaged over 100 samples.

\subsection{Small Real Networks}

\begin{figure}[t]
\subfigure[Uunet
network]{\includegraphics[width=1.6in]{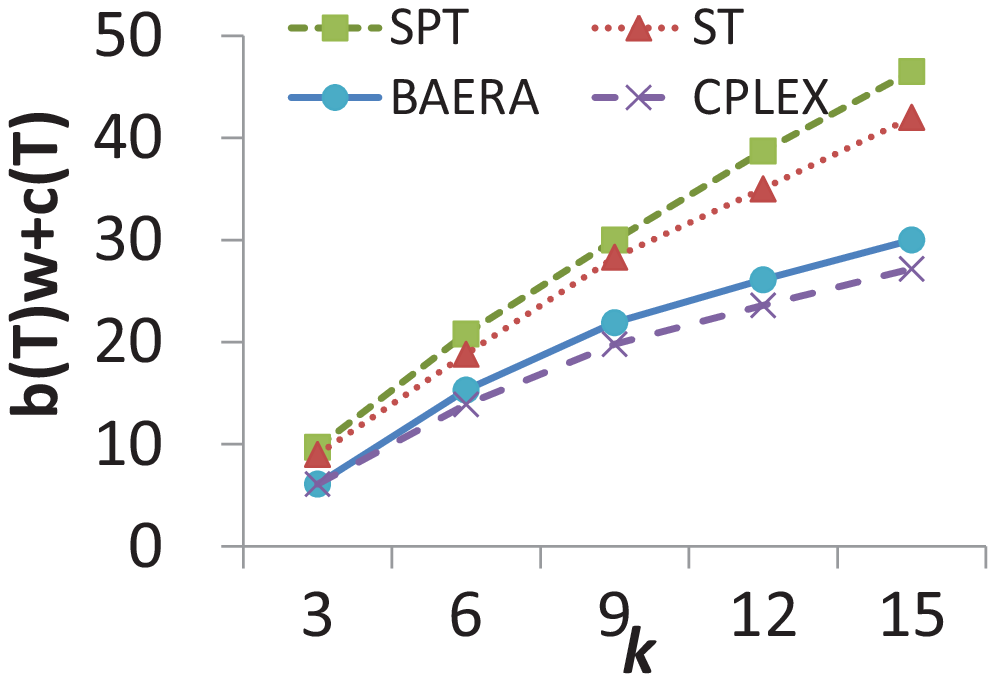}} %
\subfigure[Deltacom network]{\includegraphics[width=1.6
in]{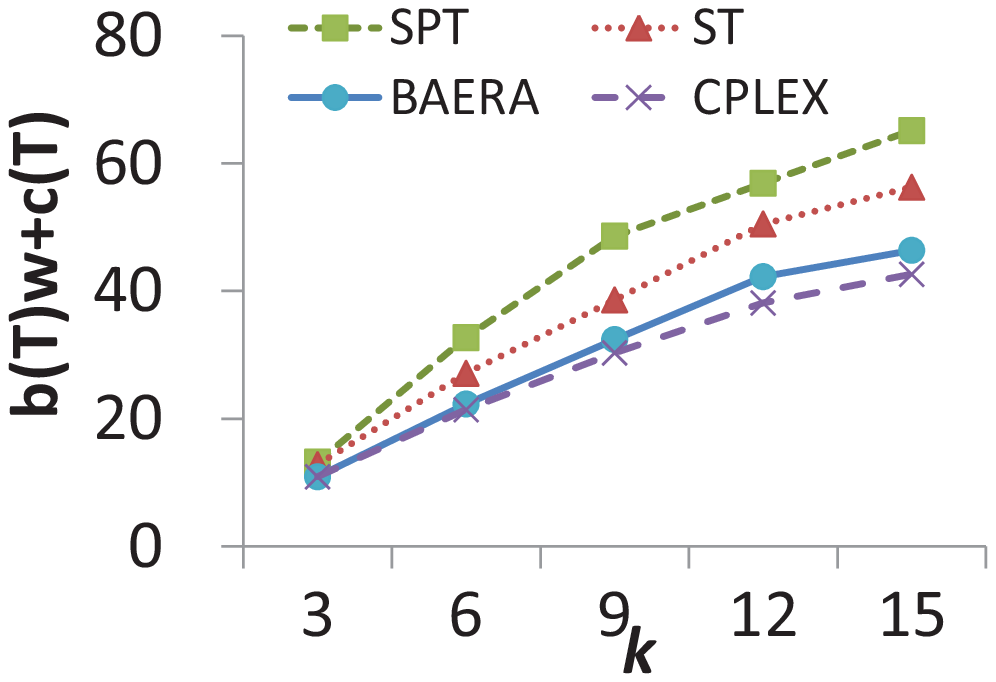}}
\caption{Varied $k$ in different real networks ($w=5$)}
\label{fig:comparing_with_cplex}
\end{figure}

In this subsection, we compare the performance of BAERA, ST and SPT with the
optimal solutions obtained by CPLEX under different $k$. Since the BST
problem is NP-Hard, CPLEX is able to find the optimal solutions for small
instances of the BST problem, and thus we only find the optimal solutions
for the Uunet and Deltacom networks. As shown in Fig.~\ref%
{fig:comparing_with_cplex}, the tree $T$ grows and includes more branch
nodes as $k$ increases, because a network is inclined to generate a large
tree. Nevertheless, BAERA outperforms SPT and ST in the two networks since
both the edge number and the branch node number are effectively minimized.
 In addition, the solutions of BAERA are very close to the optimal solutions.

\subsection{Large Synthetic Networks}

\begin{figure}[t]
\subfigure[Objective value in various $w$]
{\includegraphics[width=1.6in]{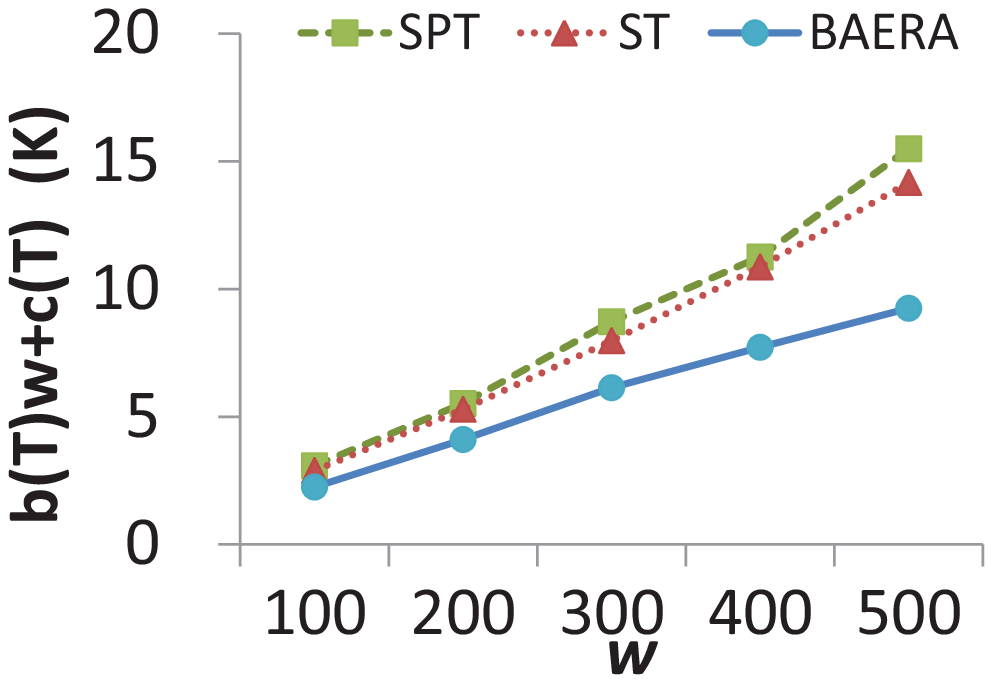}} \subfigure[Number of branch
nodes in various $w$]{\includegraphics[width=1.6 in]{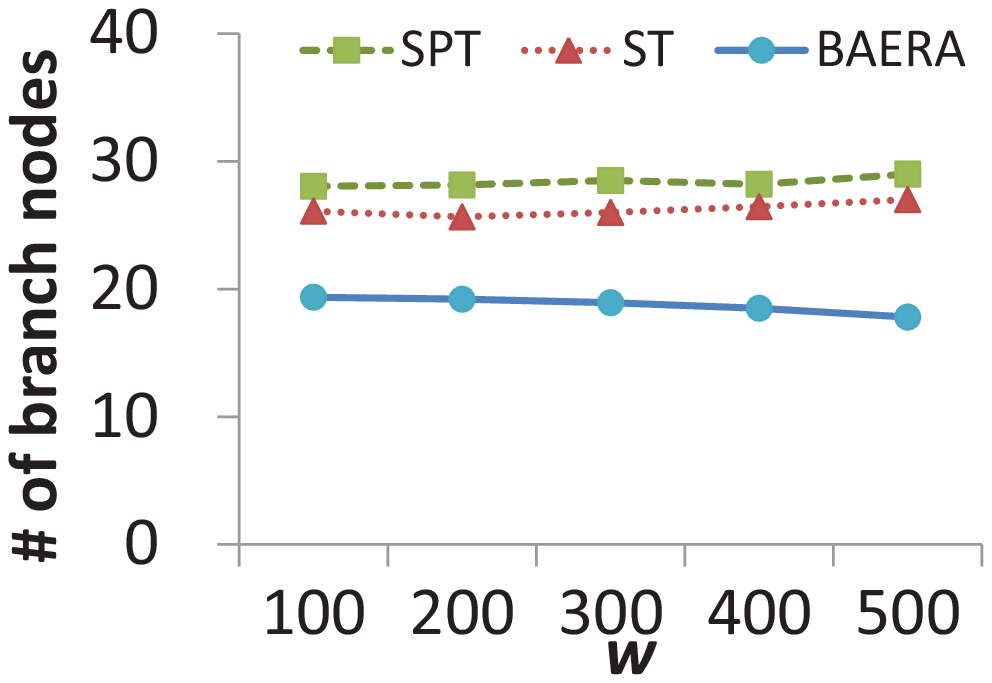}} \newline
\subfigure[Number of edges in various $w$]{\includegraphics[width=1.6
in]{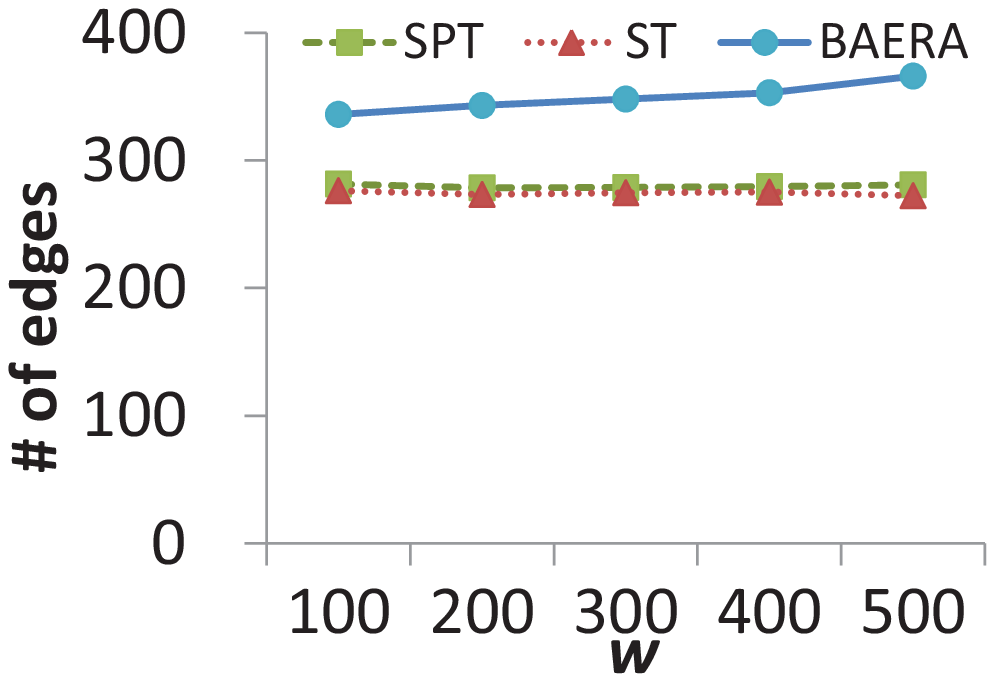}} \subfigure[Objective value in various
$k$]{\includegraphics[width=1.6 in]{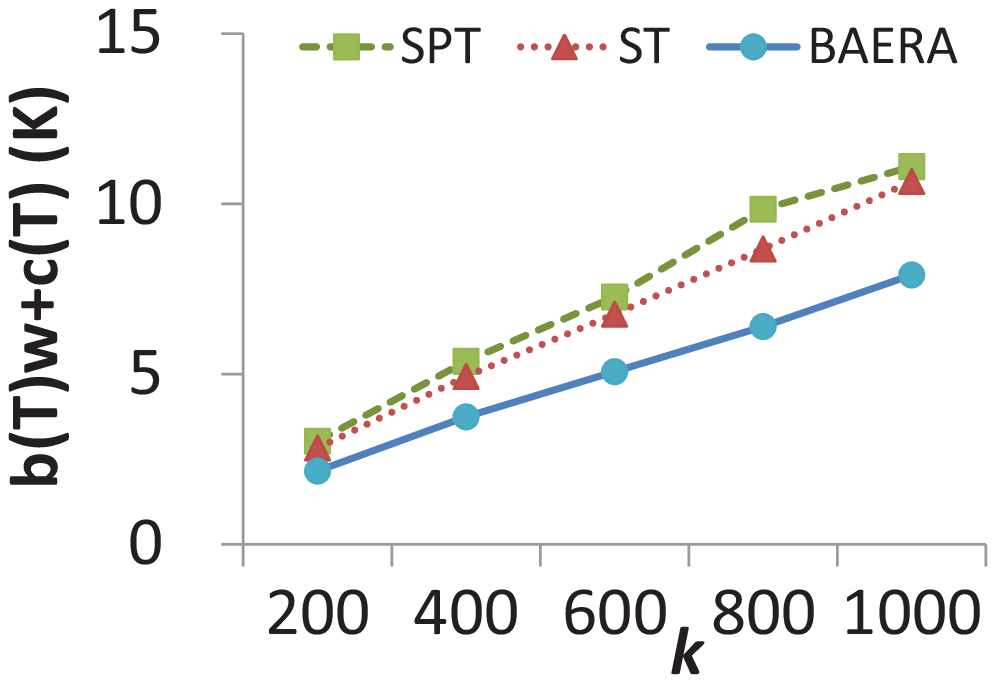}} \newline
\subfigure[Number of branch nodes in various $k$]{\includegraphics[width=1.6
in]{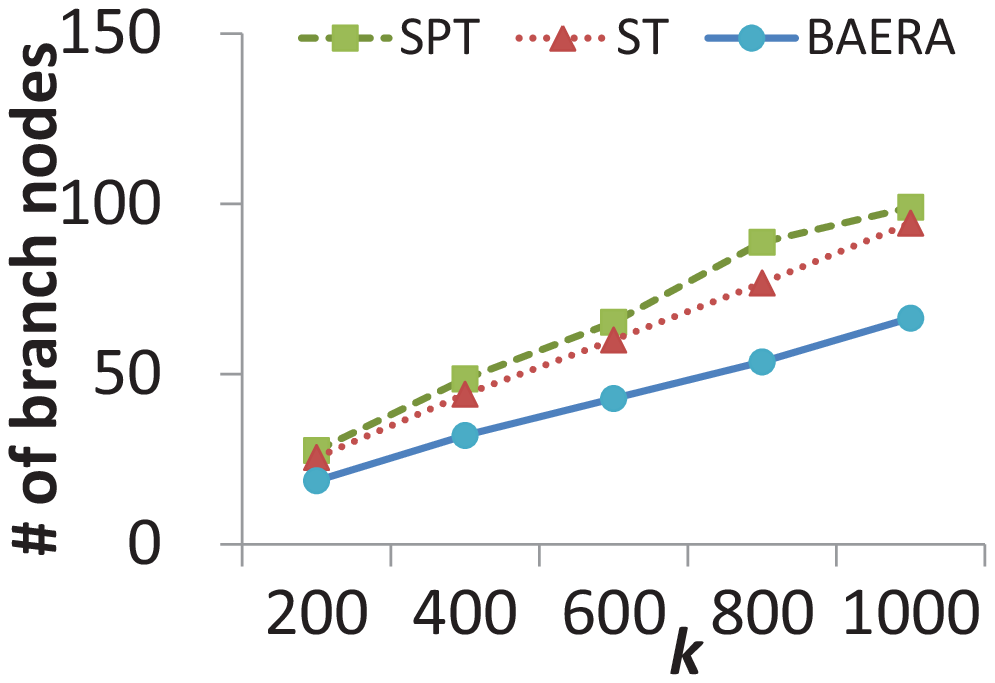}} \subfigure[Number of edges in various
$k$]{\includegraphics[width=1.6 in]{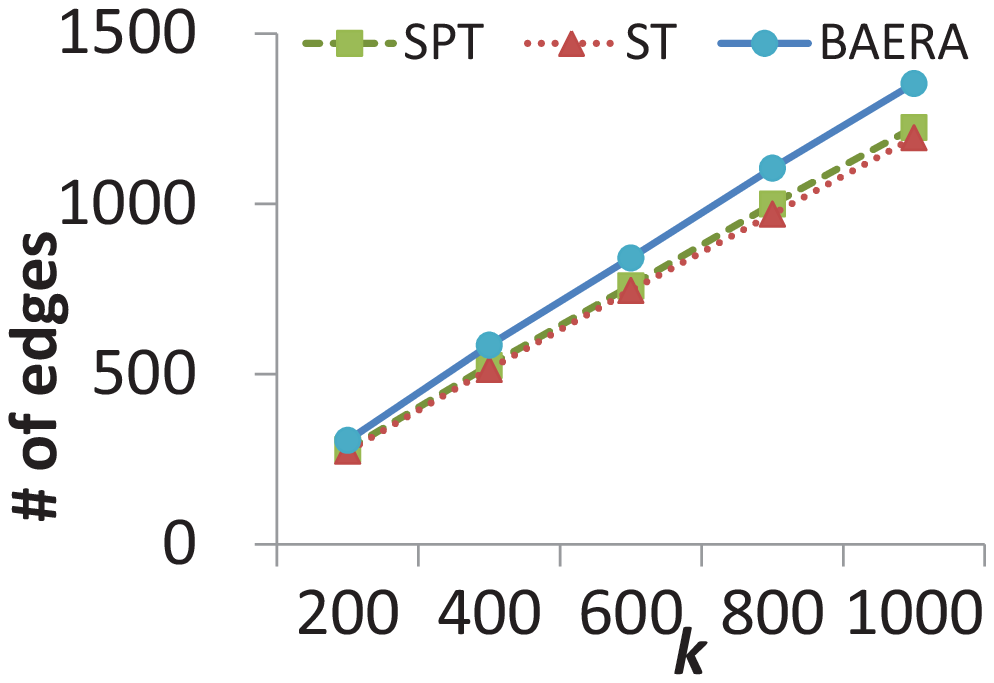}}
\caption{Varied $w$ and $k$ in the synthetic network by Inet}
\label{fig:synthetic_w_k}
\end{figure}

In the following, we evaluate BAERA, ST and SPT in large networks with 10000
nodes generated by Inet. Fig.~\ref{fig:synthetic_w_k}(a), Fig.~\ref%
{fig:synthetic_w_k}(b), and Fig.~\ref{fig:synthetic_w_k}(c) first discover
the impact of $w$ with $k$ as 200. Fig.~\ref{fig:synthetic_w_k}(a)
demonstrates that the objective value $b(T)w+c(T)$ increases as $w$ grows in
all algorithms. For a larger $w$, BST with BAERA can effectively limit the number of the created branch
nodes by slightly increasing more edges necessarily included to span all terminal nodes in $K$. Nevertheless, BAERA outperforms
SPT and ST, especially for a large $w$, because SPT and ST focus on only the
edge number and thus tend to create a tree with more branch nodes. By
contrast, the number of branch nodes in the solutions obtained by BAERA is
much smaller, but the edge number of BAERA is very close to ST.

Fig.~\ref{fig:synthetic_w_k}(d), Fig.~\ref{fig:synthetic_w_k}(e), and Fig.~%
\ref{fig:synthetic_w_k}(f) evaluate the impact of $k$ with $w$ as 100. As
shown in Fig.~\ref{fig:synthetic_w_k}(d), the objective value $b(T)w+c(T)$
becomes larger as $k$ increases, since more branch nodes are necessary to
participate in the tree. BAERA still requires fewer branch nodes from Fig.~\ref{fig:synthetic_w_k}(e). Moreover,
the increment of the objective value in BAERA grows slower than ST and BT,
showing that BAERA can further reduce the total cost in a larger $k$ with the proposed optimization methods.

Table I and Table II evaluate the running time of BAERA with various $k$ and
different Inet graph sizes. The running time of BAERA is too small to be
measured in the Uunet and Deltacom networks for arbitrary $k$. It
demonstrates that the running time of BAERA only slightly grows for a larger $%
k$, and most instances can be solved around 6 seconds when the network has
10000 nodes. In addition, for a smaller graph, ex. 4000 nodes, BAERA takes
only 1 second. Therefore, BAERA can both achieve a performance bound
(i.e., $k$-approximation) in theory and find a good solution with small time
in practice.

\begin{table}[t]
\caption{The running time of BAERA in different $k$ ($|V|$=10000)}%
\centering
\label{table:time_k}
\begin{tabular}{|c|c|c|c|c|}
\hline
$k$ & 100 & 200 & 300 & 400 \\ \hline
Running time (sec.) & 6.064 & 6.418 & 6.816 & 7.430 \\ \hline
\end{tabular}%
\end{table}

\begin{table}[t]
\caption{The running time of BAERA in different graph sizes ($k=200$)}%
\centering
\label{table:time_v}
\begin{tabular}{|c|c|c|c|c|}
\hline
$|V|$ & 4000 & 6000 & 8000 & 10000 \\ \hline
Running time (sec.) & 1.216 & 2.422 & 4.148 & 6.362 \\ \hline
\end{tabular}%
\end{table}

\section{Conclusions}
Traffic engineering and flow table scalability have been studied for unicast
traffic in SDN, but those issues in multicast SDN have not been carefully
addressed. In this paper, therefore, we exploited the branch forwarding
technique and proposed Branch-aware Steiner Tree (BST) for SDN. The BST
problem is more difficult since it needs to jointly minimize the edge and
branch node numbers in a tree, and we proved that this problem is NP-Hard
and inapproximable within $k$. To solve this problem, we designed a $k$%
-approximation algorithm, named Branch Aware Edge Reduction Algorithm
(BAERA). Simulation results manifest that the trees obtained by BAERA
include fewer edges and branch nodes, compared to the shortest-path trees
and Steiner trees. In addition, BAERA is efficient to be deployed in SDN because it can
generate a scalable and bandwidth-efficient multicast tree in massive
networks with only a few seconds.

\linespread{0.98}
\bibliographystyle{IEEEtran}
\bibliography{IEEEabrv,reference}

\begin{thebibliography}{10}
\providecommand{\url}[1]{#1}
\csname url@samestyle\endcsname
\providecommand{\newblock}{\relax}
\providecommand{\bibinfo}[2]{#2}
\providecommand{\BIBentrySTDinterwordspacing}{\spaceskip=0pt\relax}
\providecommand{\BIBentryALTinterwordstretchfactor}{4}
\providecommand{\BIBentryALTinterwordspacing}{\spaceskip=\fontdimen2\font plus
\BIBentryALTinterwordstretchfactor\fontdimen3\font minus
  \fontdimen4\font\relax}
\providecommand{\BIBforeignlanguage}[2]{{%
\expandafter\ifx\csname l@#1\endcsname\relax
\typeout{** WARNING: IEEEtran.bst: No hyphenation pattern has been}%
\typeout{** loaded for the language `#1'. Using the pattern for}%
\typeout{** the default language instead.}%
\else
\language=\csname l@#1\endcsname
\fi
#2}}
\providecommand{\BIBdecl}{\relax}
\BIBdecl

\bibitem{sdnwebsite}
\BIBentryALTinterwordspacing
Software-defined networking ({SDN}) definition. [Online]. Available:
  \url{https://www.opennetworking.org/sdn-resources/sdn-definition}
\BIBentrySTDinterwordspacing

\bibitem{McKeown2008}
N.~McKeown, T.~Anderson, H.~Balakrishnan, G.~Parulkar, L.~Peterson, J.~Rexford,
  S.~Shenker, and J.~Turner, ``{OpenFlow}: enabling innovation in campus
  networks,'' \emph{ACM SIGCOMM Computer Communication Review}, vol.~38, no.~2,
  pp. 69--74, 2008.

\bibitem{OpenFlow2013}
\emph{OpenFlow Switch Specification}, Open Networking Foundation Std. 1.4.0,
  Oct. 2013.

\bibitem{Agarwal2013}
S.~Agarwal, M.~Kodialam, and T.~Lakshman, ``Traffic engineering in software
  defined networks,'' in \emph{IEEE Proceedings of INFOCOM}, 2013, pp.
  2211--2219.

\bibitem{Kanizo2013}
Y.~Kanizo, D.~Hay, and I.~Keslassy, ``Palette: Distributing tables in
  software-defined networks,'' in \emph{IEEE Proceedings of INFOCOM}, 2013, pp.
  545--549.

\bibitem{Fenner2006}
B.~Fenner, M.~Handley, H.~Holbrook, and I.~Kouvelas, ``Protocol independent
  multicast - sparse mode (pim-sm): protocol specification (revised),''
  \emph{IETF RFC 4601}, Aug. 2006.

\bibitem{Cain2002}
B.~Cain, S.~Deering, I.~Kouvelas, B.~Fenner, and A.~Thyagarajan, ``Internet
  group management protocol, version 3,'' \emph{IETF RFC 3376}, Oct. 2002.

\bibitem{Takahashi1980}
H.~Takahashi and A.~Matsuyama, ``An approximate solution for the {Steiner}
  problem in graphs,'' \emph{Mathematica Japonicae}, vol.~24, pp. 571--577,
  1980.

\bibitem{Yang2008}
D.-N. Yang and W.~Liao, ``Protocol design for scalable and adaptive multicast
  for group communications,'' in \emph{IEEE International Conference on Network
  Protocols}, 2008, pp. 33--42.

\bibitem{YangLiao2008}
------, ``Optimal state allocation for multicast communications with explicit
  multicast forwarding,'' \emph{IEEE Transactions on Parallel and Distributed
  Systems}, vol.~19, no.~4, pp. 476--488, 2008.

\bibitem{Tian1998}
J.~Tian and G.~Neufeld, ``Forwarding state reduction for sparse mode multicast
  communication,'' in \emph{IEEE Proceedings of INFOCOM}, 1998, pp. 711--719.

\bibitem{Stoica2000}
I.~Stoica, T.~Ng, and H.~Zhang, ``Reunite: A recursive unicast approach to
  multicast,'' in \emph{IEEE Proceedings of INFOCOM}, 2000, pp. 1644--1653.

\bibitem{Wong2000}
T.~Wong and R.~Katz, ``An analysis of multicast forwarding state scalability,''
  in \emph{IEEE Proceedings of International Conference on Network Protocols},
  2000, pp. 105--115.

\bibitem{Robins2000}
G.~Robins and A.~Zelikovsky, ``Improved {Steiner} tree approximation in
  graphs,'' in \emph{Proceedings of the eleventh annual ACM-SIAM symposium on
  Discrete Algorithms}, 2000, pp. 770--779.

\bibitem{Sushant2013}
J.~Sushant, K.~Alok, M.~Subhasree, O.~Joon, P.~Leon, S.~Arjun, V.~Subbaiah,
  W.~Jim, Z.~Junlan, Z.~Min, Z.~Jon, H.~Urs, S.~Stephen, and V.~Amin, ``B4:
  experience with a globally-deployed software defined wan,'' \emph{ACM SIGCOMM
  Computer Communication Review}, vol.~43, no.~4, pp. 3--14, 2013.

\bibitem{Qazi2013}
Z.~A. Qazi, C.~C. Tu, L.~Chiang, R.~Miao, V.~Sekar, and M.~Yu, ``Simple-fying
  middlebox policy enforcement using sdn,'' \emph{ACM SIGCOMM Computer
  Communication Review}, vol.~43, no.~4, pp. 27--38, 2013.

\bibitem{Mueller2013}
J.~Mueller, A.~Wierz, and T.~Magedanz, ``Scalable on-demand network management
  module for software defined telecommunication networks,'' in \emph{IEEE
  Proceedings of SDN for Future Networks and Services}, Nov. 2013, pp. 1--6.

\bibitem{Lee2013}
B.-S. Lee, R.~Kanagavelu, and K.~M.~M. Aung, ``An efficient flow cache
  algorithm with improved fairness in software-defined data center networks,''
  in \emph{IEEE 2nd International Conference on Cloud Networking}, 2013, pp.
  18--24.

\bibitem{YangL07IPDS}
D.-N. Yang and W.~Liao, ``On bandwidth-efficient overlay multicast,''
  \emph{IEEE Transactions on Parallel and Distributed Systems}, vol.~18,
  no.~11, pp. 1503--1515, 2007.

\bibitem{Aharoni1998}
E.~Aharoni and R.~Cohen, ``Restricted dynamic {Steiner} trees for scalable
  multicast in datagram networks,'' \emph{IEEE/ACM Transactions on Networking},
  vol.~6, no.~3, pp. 286--297, 1998.

\bibitem{DBLP}
L.-H. Huang, H.-J. Hung, C.-C. Lin, and D.-N. Yang, ``Scalable {Steiner} tree
  for multicast communications in software-defined networking,'' \emph{CoRR},
  vol. abs/1404.3219, 2014.

\bibitem{Ore1960}
O.~Ore, ``Note on hamilton circuits,'' \emph{American Mathematical Monthly},
  vol.~67, no.~1, p.~55, 1960.

\bibitem{Khuller09ICALP}
S.~Khuller and B.~Saha, ``On finding dense subgraphs,'' in \emph{ICALP}, 2009,
  pp. 597--608.

\bibitem{Seidman1983}
S.~B. Seidman, ``Network structure and minimum degree,'' \emph{Social
  Networks}, vol.~5, no.~3, pp. 269--287, 1983.

\bibitem{Palmer1997}
E.~Palmer, ``The hidden algorithm of ore's theorem on hamiltonian cycles,''
  \emph{Computers \& Mathematics with Applications}, vol.~34, no.~11, pp. 113
  -- 119, 1997.

\bibitem{Zoowebsite}
\BIBentryALTinterwordspacing
The internet topology zoo. [Online]. Available:
  \url{http://www.topology-zoo.org/dataset.html}
\BIBentrySTDinterwordspacing

\bibitem{Tangmunarunkit2002}
H.~Tangmunarunkit, R.~Govindan, S.~Jamin, S.~Shenker, and W.~Willinger,
  ``Network topology generators: degree-based vs. structural,'' \emph{ACM
  SIGCOMM Computer Communication Review}, vol.~32, no.~4, pp. 147--159, 2002.

\bibitem{Inetwebsite}
\BIBentryALTinterwordspacing
Inet topology generator. [Online]. Available:
  \url{http://topology.eecs.umich.edu/inet/}
\BIBentrySTDinterwordspacing

\bibitem{cplexwebsite}
\BIBentryALTinterwordspacing
{IBM ILOG CPLEX}. [Online]. Available:
  \url{http://www-01.ibm.com/software/commerce/optimization/cplex-optimizer/}
\BIBentrySTDinterwordspacing

\end{thebibliography}

\end{document}